\documentclass[a4paper,11pt,nofootinbib]{revtex4}

\usepackage{amsfonts}
\usepackage{amssymb}
\usepackage{amsmath}
\usepackage{amsthm}
\usepackage{graphicx}
\usepackage{bbm}
\usepackage{color}
\usepackage{comment}

\newcommand*{\dd}{\textrm{d}}
\newcommand*{\prt}{\partial}
\newcommand*{\be}{\begin{equation}}
\newcommand*{\ee}{\end{equation}}
\newcommand*{\bea}{\begin{eqnarray}}
\newcommand*{\eea}{\end{eqnarray}}
\newcommand*{\ba}{\begin{align}}
\newcommand*{\ea}{\end{align}}
\newcommand*{\bd}{\begin{displaymath}}
\newcommand*{\ed}{\end{displaymath}}
\newcommand*{\nn}{\nonumber\\}

\newcommand*{\mt}{\textrm}
\newcommand*{\tr}{\textrm{tr}}

\newcommand*{\1}{\mathbbm{1}}

\newcommand*{\G}{\mathcal{G}}
\newcommand*{\g}{\mathfrak{g}}

\theoremstyle{definition}\newtheorem{definition}{Definition}
\theoremstyle{definition}\newtheorem{proposition}{Proposition}

\begin{document}

\title{Group Fourier transform and the phase space path integral\\for finite dimensional Lie groups}

\author{Matti Raasakka}
\email{matti.raasakka@aei.mpg.de}
\affiliation{Max Planck Institute for Gravitational Physics (Albert Einstein Institute), Am M\"uhlenberg 1, D-14476 Golm, Germany}

\date{\today}

\begin{abstract}
We formulate a notion of group Fourier transform for a finite dimensional Lie group. The transform provides a unitary map from square integrable functions on the group to square integrable functions on a non-commutative dual space. We then derive the first order phase space path integral for quantum mechanics on the group by using a non-commutative dual space representation obtained through the transform. Possible advantages of the formalism include: (1) The transform provides an alternative to the spectral decomposition via representation theory of Lie groups and the use of special functions. (2) The non-commutative dual variables are physically more intuitive, since despite the non-commutativity they are analogous to the corresponding classical variables. The work is expected, among other possible applications, to allow for the metric representation of Lorentzian spin foam models in the context of quantum gravity.
\end{abstract}

\maketitle

\section{Introduction}
Recently, in the context of 3d quantum gravity models, the so-called group Fourier transform was introduced \cite{FL,FM}, which provides a unitary map from functions on $SO(3)$ to functions on a dual non-commutative space, analogous to the usual Fourier transform for functions on $\mathbb{R}^d$. It was further generalized to the case of $SU(2)$ in \cite{JMN,DGL} in two different ways. (See also \cite{S,MS}.) In \cite{BDOT} the transform was used to formulate a non-commutative metric representation for Loop Quantum Gravity quantum states, dual to the usual group representation, and similarly in \cite{BO} a metric formulation of the Boulatov model for 3d Euclidean quantum gravity was derived using the transform. The metric variables obtained through the transform clarify the geometric content of these models, and in \cite{BGO} this proved to be advantageous in writing down the action of diffeomorphisms for the Boulatov model. Furthermore, in \cite{OR} the transform was applied to write down a first order phase space path integral for quantum mechanics on $SO(3)$, which was checked to agree with the well-established results obtained by more conventional methods. Here, we generalize the formalism of group Fourier transform to an arbitrary finite dimensional Lie group $\G$, and give a dual representation of quantum mechanics on $\G$ in terms of non-commutative variables. We further derive the first order phase space path integral for a quantum system on $\G$.

\section{Group Fourier transform}
Let $\G$ be a finite dimensional Lie group and $\g^*$ the dual of the Lie algebra $\g$ of $\G$. The cotangent bundle $\mathcal{T}^*\G$ of $\G$ is globally trivial, $\mathcal{T}^*\G = \G \times \g^*$, due to the automorphisms of $\G$ given by the group multiplication. The general idea behind the group Fourier transform \cite{FL,FM,JMN,DGL,OR} is the following.
\begin{definition}
The non-commutative plane wave is a function $E: \mathcal{T}^*\G = \G \times \g^* \rightarrow \mathbb{C}$, $(g,X) \mapsto E_g(X)$, from the cotangent bundle $\mathcal{T}^*\G$ of $\G$ to $\mathbb{C}$ in terms of which the delta distribution $\delta_e$ at the identity $e \in \G$ can be decomposed as
\be\label{eq:deltadecomp}
	\delta_{e}(g) = \int_{\g^*} \frac{\dd X}{(2\pi)^d}\ E_g(X)\ ,
\ee
where $\dd X$ is the Lebesgue measure on $\g^*$ and $d := \dim\G$.
\end{definition}
Now, we want to introduce a $\star$-product associated to the non-commutative plane wave $E$, which reflects the group structure. In general, it is a deformation of the usual point-wise product of functions on $\g^*$. This will later give rise to the non-commutative structure of the dual space of functions on $\g^*$ by linearity.
\begin{definition}
The $\star$-product associated to $E$ for functions on $\g^*$ is such that
\be\label{eq:stardef}
	E_g(X) \star E_h(X) = E_{gh}(X) \ .
\ee
\end{definition}
\begin{proposition}
The $\star$-product is associative, but not commutative, in general. The constant unit function $1$ is the unity with respect to the $\star$-product, if $E_{e}(X) = 1$.
\end{proposition}
\begin{proof}
The associativity and commutativity properties of the $\star$-product follow directly from the group structure, which the $\star$-product reflects. Moreover,
\be
	1 \star E_g(X) = E_e(X) \star E_g(X) = E_{eg}(X) = E_g(X) \ ,
\ee
and similarly for $E_g(X) \star 1 = E_g(X)$.
\end{proof}
From hereon we will assume that $E_{e}(X) = 1$.
\begin{definition}
The group Fourier transform $\tilde{\phi} \in L^2_\star(\g^*,\dd X)$ of a function $\phi \in L^2(\G,\dd g)$ is
\be
	\tilde{\phi}(X) := \int_\G \dd g\ E_{g^{-1}}(X) \phi(g)\ ,
\ee
where $\dd g$ denotes the left-invariant Haar measure on $\G$, and the $\star$-product is extended to $L^2_\star(\g^*,\dd X)$ by linearity.
\end{definition}
Here, the function space $L^2_\star(\g^*,\dd X)$ is defined as the image of $L^2(\G,\dd g)$ under the transform, equipped with the norm
\be
	\langle \tilde{\phi} | \tilde{\phi'} \rangle := \int_{\g^*} \frac{\dd X}{(2\pi)^d}\ \overline{\tilde{\phi}(X)} \star \tilde{\phi'}(X) \ .
\ee
\begin{proposition}
The inverse transform is obtained as
\be
	\phi(g) = \int_{\g^*} \frac{\dd X}{(2\pi)^d}\ E_g(X) \star \tilde{\phi}(X)\ .
\ee
\end{proposition}
\begin{proof}
We have
\bea
	&& \int_{\g^*} \frac{\dd X}{(2\pi)^d}\ E_g(X) \star \tilde{\phi}(X) = \int_\G \dd h \int_{\g^*} \frac{\dd X}{(2\pi)^d}\ E_g(X) \star E_{h^{-1}}(X) \phi(h) \nn
	&=& \int_\G \dd h \int_{\g^*} \frac{\dd X}{(2\pi)^d}\ E_{gh^{-1}}(X) \phi(h) = \int_\G \dd h\ \delta_g(h) \phi(h) = \phi(g) \ ,
\eea
where we used the properties (\ref{eq:deltadecomp}) and (\ref{eq:stardef}).
\end{proof}
\begin{proposition}
The transform is unitary, i.e.,
\be
	\int_\G \dd g\ \overline{\phi(g)} \phi'(g) = \int_{\g^*} \frac{\dd X}{(2\pi)^d}\ \overline{\tilde{\phi}(X)} \star \tilde{\phi'}(X) \ ,
\ee
if $E_{g^{-1}}(X) = \overline{E_{g}(X)}$.
\end{proposition}
\begin{proof}
The proof follows by a direct calculation,
\bea
	&& \int_{\g^*} \frac{\dd X}{(2\pi)^d}\ \overline{\tilde{\phi}(X)} \star \tilde{\phi'}(X) = \int_\G \dd g \int_\G \dd h \int_{\g^*} \frac{\dd X}{(2\pi)^d}\ \overline{\phi(g)}\ \overline{E_g(X)} \star E_h(X) \phi'(h) \nn
	&=& \int_\G \dd g \int_\G \dd h\ \overline{\phi(g)} \phi'(h) \int_{\g^*} \frac{\dd X}{(2\pi)^d}\ E_{g^{-1}h}(X) = \int_\G \dd g \int_\G \dd h\ \overline{\phi(g)} \phi'(h) \delta_g(h) \nn
	&=& \int_\G \dd g\ \overline{\phi(g)} \phi'(g) \ . \quad
\eea
\end{proof}
From hereon we will assume $E_{g^{-1}}(X) = \overline{E_{g}(X)}$.

It is also easy to show that $\delta_\star(X) := \int_\G \dd g\ E_g(X)$ acts as the delta distribution with respect to the $\star$-product:
\be
	\int_{\g^*} \frac{\dd X}{(2\pi)^d}\ \delta_\star(X) \star \tilde{\phi}(X) = \tilde{\phi}(0) = \int_{\g^*} \frac{\dd X}{(2\pi)^d}\ \tilde{\phi}(X) \star \delta_\star(X)\ .
\ee
In general, $\delta_\star$ is a distribution, but if $\G$ is compact, then it is in fact a regular function. 

Moreover, we will require of $E$ the following \emph{linearity} properties\footnote{The spinorial non-commutative plane waves for the group Fourier transform on $SU(2)$ used in \cite{DGL} do not satisfy the second one of these properties, and therefore it seems not possible in that case to find a simple differential group element operator in the dual representation as below.}:
\be
	\left. -i \mathcal{L}_i E_g(X) \right|_{g=e} = X_i \quad\mt{and}\quad -i\prt_X^i E_g(X) = Z^i(g) E_g(X) \ , 
\ee
where $X_i$ are Euclidean coordinates in $\g^*$, $\prt_X^i := \prt/\prt X_i$, $\mathcal{L}_i$ are left-invariant Lie derivatives on $\G$ with respect to an orthonormal basis of the Lie algebra, and $Z^i$ is a set of coordinate functions on $\G$ such that $\mathcal{L}_i Z^j(e) = \delta_i^j$ and $Z^i(g^{-1}) = -Z^i(g)$. It follows that $E_g(X) = \exp(iZ(g)\cdot X)$.

In fact, if $\G$ contains compact subgroups, we have to exclude a set of measure zero, the non-unital involutive elements, i.e., $g \in \G$ such that $g^{-1}=g$, $g\neq e$, from the range of the coordinates $Z^i$ on $\G$ due to $Z^i(g^{-1})=-Z^i(g)$ following from the requirement of unitarity for the transform. However, this is fine, since we only need (\ref{eq:stardef}) to be well-defined under integration for the group Fourier transform. Denote $\mathcal{H} := \G\backslash\{g \in \G\ |\ g^{-1}=g,g\neq e \}$. Then we may write for $\phi \in L^{2}(\G,\dd g)$
\bea
	\tilde{\phi}(X) & \equiv & \int_\G \dd g\ E_{g^{-1}}(X) \phi(g) = \int_\mathcal{H} \dd g\ E_{g^{-1}}(X) \phi(g) \nn
	&=& \int_{\mathcal{R}} \omega(Z)\dd Z\ e^{iZ(g)\cdot X} \hat{\phi}(Z) \ ,
\eea
where $\mathcal{R}\subseteq \mathbb{R}^{\dim\G}$ is the range of the coordinates $(Z^i)$ on $\mathcal{H}$, $\omega(Z)\dd Z$ is the left-invariant Haar measure in terms of the Lebesgue measure $\dd Z$ on $\mathcal{R}$, and $\hat{\phi} := \phi \circ Z^{-1}$, $Z^{-1}: \mathcal{R} \rightarrow \mathcal{H}$ being the inverse of the coordinates $(Z^i): \mathcal{H} \rightarrow \mathcal{R} \subseteq \mathbb{R}^{\dim\G}$.

The linearity properties also imply that $E_g(X)$ is the generating function for polynomials in the coordinates. In particular,
\bea
	(-i)^n \mathcal{L}_{i_1} \mathcal{L}_{i_2} \cdots \mathcal{L}_{i_n} E_e(X) &=& (-i)^n \mathcal{L}_{i_1} E_e(X) \star \mathcal{L}_{i_2} E_e(X) \star \cdots \star \mathcal{L}_{i_n} E_e(X) \nn
	&=& X_{i_1} \star X_{i_2} \star \cdots \star X_{i_n} \ ,
\eea
and so we have $[X_i,X_j]_\star := X_i \star X_j - X_j \star X_i = -i c_{ij}^{\phantom{ij}k}X_k$, where $c_{ij}^{\phantom{ij}k}$ are the structure constants of $\g$. Accordingly, the non-commutative $\star$-product structure reflects the structure of the Lie algebra of $\G$.
\begin{proposition}
Let $\tilde{\phi},\tilde{\phi}' \in L^{2}_\star(\g^*,\dd X)$, and $\omega(Z)\dd Z$ be the Haar measure $\dd g$ on $\mathcal{H} \subseteq \G$ in terms of the Lebesgue measure $\dd Z$ on $\mathcal{R} \subseteq \mathbb{R}^{\dim\G}$. Then
\be\label{eq:starint}
	\int_{\g^*} \frac{\dd X}{(2\pi)^d}\ \tilde{\phi}(X) \star \tilde{\phi}'(X) \equiv \int_{\g^*} \frac{\dd X}{(2\pi)^d}\ \tilde{\phi}(X)\ \omega^{-1}\!\!\left(-i\prt_X^i\right) \tilde{\phi}'(X) \ .
\ee
\end{proposition}
\begin{proof}
For $g,h\in\mathcal{H}$ we have
\bea
	&& \int_{\g^*}\frac{\dd X}{(2\pi)^d}\ E_{g^{-1}}(X) \star E_h(X) \equiv \delta_g(h) \nn
	&=& \omega^{-1}\!(Z^i(h))\delta(Z(g)-Z(h)) = \omega^{-1}\!(Z^i(h)) \int_{\g^*}\frac{\dd X}{(2\pi)^d}\ e^{-iZ(g)\cdot X} e^{iZ(h)\cdot X} \nn
	&=& \int_{\g^*}\frac{\dd X}{(2\pi)^d}\ E_{g^{-1}}(X) \omega^{-1}\!(-i\prt_X^i) E_h(X) \ ,
\eea
By extending the above equality by linearity to functions $\tilde{\phi} \in L^{2}_\star(\g^*,\dd X)$ we arrive at the proposition.
\end{proof}
\begin{proposition}
If $Z(hgh^{-1})\cdot X = Z(g)\cdot Ad_h X$, we have
\be\label{eq:ad}
	E_g(X) \star \tilde{\phi}(X) \star E_{g^{-1}}(X) = \tilde{\phi}(Ad_gX) \ .
\ee
\end{proposition}
\begin{proof}
The proposition follows by linearity from
\be
	E_h(X) \star E_g(X) \star E_{h^{-1}}(X) = E_{hgh^{-1}}(X) = E_g(Ad_hX) \ . \quad
\ee
\end{proof}
We will hereon assume $Z(hgh^{-1})\cdot X = Z(g)\cdot Ad_h X$. Also, using this property, it is easy to show that
\be\label{eq:intcycl}
	\int_{\g^*} \frac{\dd X}{(2\pi)^d}\ \tilde{\phi}_1(X) \star \tilde{\phi}_2(X) \star \cdots \star \tilde{\phi}_n(X) = \int_{\g^*} \frac{\dd X}{(2\pi)^d}\ \tilde{\phi}_2(X) \star \cdots \star \tilde{\phi}_n(X) \star \tilde{\phi}_1(X) \ ,
\ee
i.e., an integral over $\g^*$ is invariant under a cyclic permutation of the $\star$-product factors of its integrand.

The choice for the explicit form of the non-commutative plane waves $E_g(X)$ is not completely obvious. It dictates the explicit form of the $\star$-product, and is thus connected to the issue of operator ordering in the quantization map, which the $\star$-product reflects from the deformation quantization point of view \cite{CD}. The canonical choice for coordinate functions $Z^i$ on the identity component of $\G$, when the exponential map is surjective, is obtained from the inverse of the exponential map $\exp: \g \simeq \mathbb{R}^d \rightarrow \G$, $Z \mapsto \exp Z$, for a portion of $\g$ containing the origin, where the exponential map is one-to-one. With this choice we have the natural inner product $\cdot: \g \times \g^* \rightarrow \mathbb{R}$ to use for the plane waves $E_g(X) = \exp(iZ(g)\cdot X)$.\footnote{In the case of several components or non-compact Lie groups for which the exponential map is not surjective, we may generalize to several coordinate patches as in \cite{JMN}, using coordinates $Z_h(g) \in \g$ such that $g\equiv he^{Z_h(g)}$ for different origins $h \in \G$ for the coordinate patches.} These coordinates satisfy all the above requirements. However, mainly other choices for the coordinates have been considered in the literature in the cases of the groups $SO(3)$ and $SU(2)$ \cite{FL,FM,JMN,DGL,BO,BDOT,BGO,BO2,OR}. In particular, \cite{FL,FM,BO,OR} use the coordinates $Z^i(g) = -\frac{i}{2}\tr(g\sigma^i)$ for $SO(3)$, where the trace is taken in the fundamental spin-$\frac{1}{2}$-representation and $\sigma^i$ are the Pauli matrices. They satisfy all the above assumptions for $Z^i$. These coordinates may be more convenient from the calculational point of view, but are not generalizable to arbitrary Lie groups.

Let us summarize the properties of the non-commutative plane wave, i.e., the function $E:\mathcal{T}^*\G \rightarrow \mathbb{C}$ we will use in the following. In short, we assume
\be
	E_g(X) \equiv \exp(iZ(g)\cdot X) \ ,
\ee
where $Z^i(g)$ are coordinates on $\G\backslash\{g \in \G\ |\ g^{-1}=g,g \neq e\}$ such that
\begin{enumerate}
	\item $Z^i(g^{-1})=-Z^i(g)$
	\item $\mathcal{L}_iZ^j(e)=\delta_i^j$
	\item $Z(hgh^{-1})\cdot X = Z(g)\cdot Ad_hX$.
\end{enumerate}
With these assumptions the above construction for the group Fourier transform applies.

\section{Dual representation of quantum mechanics on $\G$}
To describe a quantum system with $\G$ as its configuration space, we start as usual with the Hilbert space $L^2(\G,\dd g)$ of wave functions on $\G$. Employing the Dirac notation, a basis $\{|g\rangle\ |\ g \in \G\}$ is given by delta distributions $\langle h | g \rangle := \delta_g(h)$. We may understand $|g\rangle$ as eigenstates of the group element operator $\hat{g}$ in the abstract sense $\phi(\hat{g})|g\rangle \equiv \phi(g)|g\rangle$, where $\phi$ is a function on $\G$. In particular, we denote $\hat{Z}^i := Z^i(\hat{g})$. We may define operators $\hat{X}_i := -i\mathcal{L}_i$, which are conjugate to $\hat{Z}^i$ in the sense of $[\hat{X}_i,\hat{Z}^j]|e\rangle = -i\delta_i^j|e\rangle$, i.e., they satisfy the canonical commutation relations at the unit element (but receive non-linear corrections away from the unity).\footnote{We have set $\hbar=1$.} These commutators reproduce the canonical Poisson algebra of the cotangent bundle\footnote{The classical Poisson algebra follows unambiguously from the canonical symplectic structure of $\mathcal{T}^*\G$.} $\mathcal{T}^*\G \simeq \G \times \g^*$ of $\G$, given by the Poisson bracket
\be
	\{F,G\}_{PB} = (\prt_X^i F) (\mathcal{L}_i G) - (\prt_X^i G) (\mathcal{L}_i F) + c_{ij}^{\phantom{ij}k} X_k (\prt_X^i F) (\prt_X^j G)\ ,
\ee
in the appropriate way, namely, $\widehat{\{F,G\}}_{PB} = i[\hat{F},\hat{G}]$ for the canonical operators.

Now, we may consider another set of state vectors $\{|X\rangle\ |\ X \in \g^*\}$, defined via $\langle g | X \rangle := E_g(X)$. These states also form a basis of the Hilbert space in the sense of the $\star$-product, 
\be
	\langle Y | X \rangle \equiv \delta_\star(X-Y) \quad\mt{and}\quad \int_{\g^*} \frac{\dd X}{(2\pi)^d}\ |X\rangle\star\langle X| \equiv \hat{\1}\ ,
\ee
and therefore we obtain a dual representation of the quantum system in terms of $X \in \g^*$:
\be
	\langle X | \phi \rangle = \int_\G \dd g\ \langle X | g \rangle \langle g | \phi \rangle \ .
\ee
The canonical operators are given in this basis by
\be
	\hat{Z}^i|X\rangle = -i\prt_X^i|X\rangle \quad\mt{and}\quad \hat{X}_i|X\rangle = |X\rangle \star X_i \ .
\ee
Then, we also have $\phi(\hat{g})|X\rangle = \phi\big(g(-i\prt_X^i)\big)|X\rangle$, where $g(Z^i)$ is the inverse of the coordinates $(Z^i(g))$.\footnote{Here we must again note that, strictly speaking, $(Z^i(g))$ is not invertible for a set of measure zero in $\G$. However, apart from the abstract Dirac notation, we will only need to assume these relations to be defined under integration over $\G$, since this is how the dual representation is defined to begin with.}

\section{Phase space path integral}
Now, taking advantage of the dual representation for the Hilbert space introduced in the previous section, we may write down the first order phase space path integral for a quantum system with $\G$ as the configuration space. As usual, we have a Hamiltonian operator $\hat{H}$ generating the time-evolution of the system. Then, the propagator for the quantum system in $\G$-basis reads
\be
	\langle g',t' | g,t \rangle := \langle g'| e^{ -i(t'-t)\hat{H} } |g\rangle \ .
\ee
Introducing time-splitting $(t'-t) \equiv \epsilon N$ and inserting resolutions of identity $\hat{\1} \equiv \int_\G \dd g\ |g\rangle \langle g|$ and $\hat{\1} \equiv \int_{\g^*} \frac{\dd X}{(2\pi)^d}\ |X\rangle \star \langle X|$ for each timestep yields
\be
	\langle g',t' | g,t \rangle = \left[ \prod_{k=1}^{N-1} \int_\G \dd g_k \right] \left[ \prod_{k=0}^{N-1} \int_{\g^*} \frac{\dd X_k}{(2\pi)^d} \right] \left[ \prod_{k=0}^{N-1} \langle g_{k+1}|X_{k}\rangle \star \langle X_{k}| e^{ -i\epsilon\hat{H} } |g_{k}\rangle \right] \ .
\ee
To obtain the path integral, we take the limits $\epsilon \rightarrow 0$, $N \rightarrow \infty$, while $\epsilon N = (t'-t)$. For the path integral to satisfy the Schr\"odinger equation, only the linear order in $\epsilon$ must be taken into account \cite{DeWitt}, and we may approximate 
\be
	\langle X_{k}| e^{ -i\epsilon\hat{H} } |g_{k}\rangle \approx \big(1 - i\epsilon H_{\star}(g_k,X_k)\big) \star \langle X_{k}|g_{k}\rangle \approx e^{ -i\epsilon H_{\star}(g_k,X_k) } \star \langle X_{k}|g_{k}\rangle \ ,
\ee
where $\langle g|\hat{H}|X\rangle =: \langle g|X\rangle \star H_\star(g,X)$. Further, using (\ref{eq:starint}) and (\ref{eq:intcycl}), we arrive at the discrete timestep expression for the path integral
\bea
	\langle g',t' | g,t \rangle &=& \lim_{\substack{\epsilon \rightarrow 0\\N \rightarrow \infty}} \left[ \prod_{k=1}^{N-1} \int_\G \dd g_k \right] \left[ \prod_{k=0}^{N-1} \int_{\g^*} \frac{\dd X_k}{(2\pi)^d} \right] \nn
	&& \times \exp\left\{ i \sum_{k=0}^{N-1} \epsilon \left( \frac{1}{\epsilon}Z(g_{k}^{-1}g_{k+1}) \cdot X_{k} - H_q(g_k,X_k) \right) \right\} \ ,\label{eq:PIdisc}
\eea
where $H_q(g,X) := \omega^{-1}\!(-i\prt_X^i) H_\star(g,X)$ is \emph{the quantum corrected Hamiltonian} for the system. These quantum corrections to the classical Hamiltonian, which are due to the non-commutative structure of the phase space, are crucial for the propagator to satisfy the Schr\"odinger equation. At the continuum limit (\ref{eq:PIdisc}) can be formally expressed as a continuum path integral
\be\label{eq:PIcont}
	\langle g',t' | g,t \rangle =  \int_{\substack{g(t)=g\\g(t')=g'}} \mathcal{D}g\ \mathcal{D}X\ \exp\left\{ i \int_{t}^{t'}\!\!\! \dd t\ L_q(g,X) \right\} \ ,
\ee
where $L_q(g,X) := V \cdot X - H_q(g,X)$ is the quantum corrected Lagrangean with
\be
	V^i(t) := \lim_{\epsilon \rightarrow 0} Z^i(g^{-1}(t)g(t+\epsilon))/\epsilon \ .
\ee

For the special case of $\G = \mathbb{R}^d$ with the coordinates $Z^i(g)$ obtained via the inverse exponential map (\ref{eq:PIdisc}) and (\ref{eq:PIcont}) agree with the usual expressions for the first order path integral. In \cite{OR} the stationary phase approximation to (\ref{eq:PIcont}) was studied for $SO(3)$ with coordinates $Z^i(g)=-(i/2)\tr(g\sigma^i)$ and found to yield the correct classical equations of motion. Also, it was shown that for the case of a free particle on $SO(3)$ the quantum corrections agree with the well-established results \cite{DeWitt,MT,CD} obtained by more conventional methods.

\section{Conclusion \& Comments}
We have formulated a notion of group Fourier transform for a finite dimensional Lie group $\G$, and used it to write down the phase space path integral for a quantum system on $\G$. The main advantages of the approach to quantum systems on Lie groups presented here are:
\begin{enumerate}
	\item The non-commutative dual variables are analogous to the classical commutative ones, so their physical meaning is intuitive.
	\item There is no need to resort to the representation theory of Lie groups and special functions, which often are rather involved.
\end{enumerate}
The price to pay for these simplifications is, of course, the non-commutativity of the dual space, and depending on the system in case one may prefer one approach over the other. However, the clear physical interpretation of the non-commutative dual variables has already proven itself useful in the case of quantum gravity models in shedding light on their geometrical content \cite{BDOT,BO,BGO,BO2}. The generalization of the group Fourier transform to an arbitrary Lie group removes the restriction of these considerations to the case of products of $SO(3)$ or $SU(2)$. In particular, from the point of view of quantum gravity, the possibility to formulate the metric representation for 4d quantum gravity models using the transform for the Lorentz group is important. Considering the second point, also the study of the semi-classical limit of the quantum system is greatly facilitated by not having to deal with special functions \cite{OR}, and therefore we expect non-commutative methods to be useful also in studying the asymptotics of quantum gravity models.

\section*{Acknowledgements}
The author wishes to thank Daniele Oriti for his encouragement to write this paper, and Carlos Guedes for many useful discussions on the subject.

\end{document}